\documentclass[10pt]{amsart}
\usepackage{amssymb}
\usepackage{graphicx}

\theoremstyle{plain}
\newtheorem{thm}{Theorem}

\newtheorem{lem}[thm]{Lemma}

\newtheorem*{claim}{Claim}
\theoremstyle{definition}

\newtheorem{exm}[thm]{Example}

\renewcommand{\leq}{\leqslant}

\newcommand{\rto}{\rightarrow}

\newcommand{\Cl}{\operatorname{Cl}}
\newcommand{\Ji}{\operatorname{Ji}}
\newcommand{\Mi}{\operatorname{Mi}}

\newcommand{\ua}{\uparrow}
\newcommand{\da}{\downarrow}
\newcommand{\uda}{\updownarrow}

\begin{document}

\title[Discovery of the $D$-basis]{Discovery of the D-basis in binary tables based on hypergraph dualization}
\author {K. Adaricheva}
\address{Department of Mathematics, School of Science and Technology, Nazarbayev University, 53 Kabanbay Batyr ave., Astana 010000, Republic of Kazakhstan}
\email{kira.adaricheva@nu.edu.kz}
\address{Department of Mathematical Sciences, Yeshiva University,
245 Lexington ave.,
New York, NY 10016, USA}
\email{adariche@yu.edu}

\author {J. B. Nation}
\address{Department of Mathematics, University of Hawaii, Honolulu, HI
96822, USA}
\email{jb@math.hawaii.edu}

\thanks{The research was partially supported by Grant N 13/42 of Nazarbayev University, 2013--2105}

\begin{abstract}
Discovery of (strong) association rules, or implications, is an important task in data management, and it finds application in artificial intelligence, data mining and the semantic web. We introduce a novel approach for the discovery of a specific set of implications, called the $D$-basis,  that provides a representation for a reduced binary table, based on the structure of its Galois lattice. At the core of the method are the $D$-relation defined in the lattice theory framework, and the hypergraph dualization algorithm that allows us to effectively produce the set of transversals for a given Sperner hypergraph. The latter algorithm, first developed by specialists from Rutgers Center for Operations Research, has already found numerous applications in solving optimization problems in data base theory, artificial intelligence and game theory. One application of the method is for analysis of gene expression data related to a particular phenotypic variable, and some initial testing is done for the data provided by the University of Hawaii Cancer Center. 
\end{abstract}

\maketitle

\section{Introduction}

Knowledge retrieval from large data sets remains an essential problem in information technology and its many usages in finance, biology, economy and social sciences. The data is often recorded in binary tables with rows consisting of the objects and columns of the attributes, that mark whether a particular object has or does not have a particular attribute. The dependencies existing between the subsets of the attributes in the form of association rules, or implications, can uncover the laws, causalities and trends hidden in the data; see R. Agrawal et al.~\cite{ASMTV96}.

	An implication $X\rightarrow Y$ (also referred to as a strong association rule in data mining, see M. Kryszkiewicz \cite{K02}, with parameter of confidence equal to 1) has the meaning that every object in the data set that possesses each attribute in the subset X will also possess every attribute from the set $Y$. Such implications provide an essential hidden connection between different attributes. Sets of implications which are called bases provide an alternative way of storing the tabled data: they generate all possible implications that hold in the set of attributes, as their logical consequences, and also allow the restoration of the tabled data. Sometimes, only the implications $X\rightarrow b$ pertinent to a distinguished attribute $b$ could be of particular importance in various experiments that produce tabled data. 

One particular type of a table could be medical data where  $b$
may stand for a phenotypic attribute, while other attributes are the expression levels of specific genes. Then each association rule may represent a specific hypothesis of how the expression of several genes directly impacts the expression of another gene, and via a chain of direct interactions, indirectly impacts the phenotypic variation. Biologists may then pick some of these hypotheses to verify the interactions with the inventory of biochemical methods or animal models. To discriminate between the rules that are biologically relevant and those that are not, the concept of the support of the rule $X\rightarrow b$ may be applied, which is a portion of the rows of the table where all attributes from X are present. The rules with the highest support could be the targets of prime interest.
	
In this paper we propose a novel algorithm for the retrieval of the strong association rules between the sets of attributes in large binary tables, which is particularly suited for retrieval of targeted implications  $X\rightarrow b$ with a fixed attribute $b$. At the heart of the algorithm lies the connection between an arbitrary binary table and an associated algebraic structure known as a Galois lattice. A generating set of strong association rules, called a basis, often gives a more concise representation of the tabled data and its Galois lattice. Nevertheless, the existing algorithms developed in the framework of Formal Concept Analysis (FCA) are time-exponential in the size of the table, and they become ineffective on large data sets.

Our approach retrieves a complete generating set of implications known as the $D$-basis, that was introduced and tested in K. Adaricheva et al.~\cite{ANR11}. This new type of basis is amenable to parallel processing, and can be used to find the rules that target a particular attribute without constructing the entire Galois lattice.
Significantly, it employs the powerful optimization algorithm known as the hypergraph dualization, see M. Fredman and L. Khachiyan \cite{FK96}. Similar arrpoach was used in U.~Ryssel, F.~Distel and D.~Borchmann \cite{RDB11} for the retreival of the canonical direct unit basis, of which the $D$-basis is a subset. The code implementations of the dualization algorithm were successfully tested on the models of the large size in E. Boros et al.~\cite{BEGK02}. While one can employ the algorithm for arbitrary tables, it could especially be effective in tables with the number of attributes considerably larger than the number of the objects. One finds this type of data in biological studies, where the set of attributes may include millions of genetic variations.

The paper is organized as follows. We present the association rules and parameters of support and the confidence, as they are known in data mining, in section \ref{Asso-rul}. Then we give a general overview for the Galois lattices and FCA methods in section \ref{FCA} and the hypergraph dualization algorithm in section 
\ref{Hyper}. Section \ref{Dbas} describes the $D$-basis, and section \ref{Galois} gives the general properties of Galois lattices. Our main result is in section \ref{D-bas-table}: an algorithm to recover the $D$-basis from a binary table. We demonstrate the algorithm on an example in section \ref{Exm}, provide some preliminary results of testing the code implementation in section \ref{test}, and give an overview of the future work in section \ref{future}.

\section{Association rules in data mining}\label{Asso-rul}

The leading algorithmic approach in recovery of association rules in data mining was formulated in the form of the \emph{Apriori} algorithm in R. Agrawal et al.~\cite{ASMTV96}. Two parameters are essential for association rules: the thresholds of the support and the confidence.  The support of any set $X$ of attributes is the value $S_A (X)$ (formally defined in section \ref{Galois}), which is essentially the number of rows of the table that have ones in all columns from  $X$. The set  $X$ is called $\delta$-frequent, if the ratio of $S_A (X)$ to the number of all rows in the table is greater than the given value of the threshold $\delta$.
The first part of the \emph{Apriori} algorithm aims at discovery of the maximal $\delta$-frequent sets, for a given threshold value $\delta$.

The second part of the algorithm is devoted to splitting any maximal frequent set $X$ into $Y\cup Z$ so that the confidence $S_A (Y\cup Z) / S_A (Y)$ exceeds a given threshold $\mu$. In this case, $Y\rightarrow Z$ is an association rule that satisfies both thresholds $\delta$ and $\mu$.
 
The time required to build maximal frequent sets, in the bottom-up way as in \emph{Apriori}, is asymptotically proportional to the number of frequent sets, and the latter may be exponentially larger than the number of maximal frequent sets. The number of maximal frequent sets itself may also be exponential in the size of the attributes of the table. The time required to build maximal frequent sets can be used as a norm to evaluate the complexity relative to the size of both input and output, in other words,  the time delay in computation of the next maximal frequent set.

According to a result of E. Boros et al.~\cite{BGKM02}, given any subset  $S$ of maximal frequent sets, the problem of deciding whether it is complete (i.e., includes all maximal sets) is NP-hard, even if the number of maximal frequent sets is exponential in the number of attributes  $n$, and $S$ is of size $\mathcal{O}(n^\alpha)$ for small $\alpha$. 

All this makes the \emph{Apriori} heuristic inadequate for large data sets, in particular with many frequent sets of large size.
Another weakness of the \emph{Apriori} approach is its inability to target specific attributes in the conclusion of association rules. Due to the nature of the algorithm, the decision about the splitting of a maximal frequent set into antecedent and conclusion happens in the second part of the algorithm, after the main effort on obtaining the frequent sets is already spent. See further details in section \ref{test}.

\section{Formal Concept Analysis and basis retrieval}\label{FCA}

\subsection{Retrieval of the canonical basis}\label{can-bas}
The path from the tabled data to frequent sets of attributes, and then to bases of implications, can be built via the Galois relation existing between objects and attributes, and via the structure associated with this relation known in the algebraic literature as the Galois lattice. For example, maximal frequent subsets of the attribute set, which are targeted in data mining of transaction tables, are particular elements of the Galois lattice, see P. Valtchev et al.~\cite{VMG08}, while the strong association rules (with the confidence = 1) form the implicational basis of this lattice.

The first study of the Galois lattice appeared in 1940 in G. Birkhoff \cite{B40}. It was developed further in M. Barbut and B. Monjardet \cite{BM40}, and taken into a field of its own, under the name Formal Concept Analysis (or FCA), in B. Ganter and R. Wille \cite{GW99}. FCA took off as a technical tool since the 1980s, mostly for the data visualization for various business applications. The applications range from knowledge representation and data mining to knowledge management and the semantic web, see J. Poelmans et al.~\cite{PEVD10}. 

Most recently, there were several projects in European Commonwealth where FCA was involved in biological studies of temporal Boolean networks modeling genetic data. See, for example, \cite{AP05,D01,WHPKGKG09,WGG08}. There are growing applications of Galois lattices in ontology mappings \cite{YSH08,ZS02,YXH06} and in description logics \cite{S08}. A generalization of Galois lattices, known as ``pattern structures" in B. Ganter and S. Kuznetsov \cite{GK01}, deals with data more complex than binary tables. Some applications of these techniques can be found in C. Carpineto and G. Romano \cite{CR04} and M. Kaytoue  et al.~\cite{KKND11}.

The FCA approach is centered around a special implicational basis known in the literature as the canonical basis (also stem basis, or Duquenne-Guigues basis). This basis has the minimum number of implications, and has a fundamental connection to any other basis defining a given Galois lattice and/or binary table.

The attribute exploration algorithm  to recover the canonical basis, developed in the FCA framework, see B. Ganter \cite{G99}, requires a search over the entire Galois lattice, with respect to some linear order established on the power set of attributes. As a result, the algorithm runs in times dependent on at least the size of the Galois lattice, which is normally exponential in the number of attributes. Recent theoretical results suggest that existing approaches to calculation of the canonical basis may not produce algorithms with better worst-case complexity, see \cite{BK10,DS11}.

At its core, the implications of the canonical basis are defined by recognizing pseudo-closed sets $X$ that are placed into the premises of implications: $X\rightarrow Y$. There is no possibility to use parallel computation of pseudo-closed sets, since every new one depends on its subsets being recognized earlier in the process.\footnote{We mention that some efforts were spent on parallelization of algorithms building a Galois lattice from the binary table, such as in P.~Krajca et al.~\cite{KOV12}, but possible utilization of these approaches for computation of pseudo-closed sets has not yet been developed.}   Moreover, it may happen that the same attribute $b$ will appear on the right side in several implications from the basis. Thus, even if we were interested only in implications of the form $X\rightarrow b$, for this particular $b$, one would need to reconstruct the whole basis.

\subsection{Canonical direct unit basis}\label{Can-unit}

The canonical direct unit basis is discussed in K. Bertet and B. Monjardet \cite{BM10} as a basis unifying various definitions given to this concept in the literature. As the Duquenne-Guigues basis, it can be defined in the abstract framework of a general closure system on a finite set. It consists of implications  $X\rightarrow b$ with the property of minimality of $X$ with respect to set containment, for any fixed $b$.  In other words, no implication  $X'\rightarrow b$ holds, where $X'$ is a proper subset of $X$.

The canonical direct unit basis has considerably more implications compared to the Duquenne-Guigues basis, but it has a nice feature of being iteration-free basis, or direct basis. In fact, it is contained in any other direct basis defining the same closure system. Section 11 of \cite{ANR11} contains experimental results on the time needed to compute the closure of a randomly chosen set, in the closure systems on 6 and 7-element sets, given three bases: canonical basis of Duquenne-Guigues, in its unit form, the canonical direct unit basis and the $D$-basis.
For the closure systems in these computational experiments,
we found that the length of the Duquenne-Guigues basis was on
the order of half the length of the canonical direct basis.
On the other hand, not being direct or ordered direct, it
took on the order of twice as long to compute the closure
as the other two.


Recently, U. Russel et al.~\cite{RDB11} proposed a method of retrieval of the canonical direct unit basis that would employ the hypergraph dualization algorithm. While the actual coding applied some back-tracking and simple heuristics instead of actual implementation of dualization, the computer test results showed considerably shorter running times when compared to existing FCA implementations: on a particular data table of size $26\times 79$ the new algorithm would return the basis of 86 implications in 0.1 sec compared to 6.5 hours of the FCA implementation, and, on random tables of size $20\times 40$, it spent about 1/50 of time needed for the FCA algorithm.

We note that the $D$-basis which we deal with in this paper is a subset of the canonical direct unit basis, see more details in section \ref{Dbas} and in \cite{ANR11}.

\section{Hypergraph dualization}\label{Hyper}

Let $A$ be a finite set of cardinality $|A|=n$. For a hypergraph (set family) $H\subseteq 2^A$, consider the family of its maximal independent subsets, i.e., maximal subsets of $A$ not containing any edge of $H$. The complement of a maximal independent subset is a minimal transversal of $H$, i.e., a minimal subset of $A$ intersecting all edges of $H$.
(Minimal transversals are also called minimal hitting sets.)

The collection $H^d$ of minimal \emph{transversals} is called the dual or transversal hypergraph of $H$. It is easy to see that $H^d$ is a Sperner hypergraph, i.e., no edge of $H^d$ contains another edge of $H^d$. If $H$ is also Sperner then $H=(H^d)^d$. Given a Sperner hypergraph, a frequently arising task is the generation of the transversal hypergraph $H^d$. This problem,
known as dualization, can be stated as follows:\\

$\operatorname{DUAL}(H,G)$::Given a complete list of all edges of $H$, and a set of minimal transversals $G\subseteq  H^d$, either prove that $G=H^d$, or find a new transversal $g\in H^d\setminus G$.

Clearly, one can generate all of the minimal transversals in $H^d$ (equivalently, all the maximal independent sets for $H$) by
initializing $G=\emptyset$  and iteratively solving the above
problem $|H^d |+1$ times. Since $|H^d |$ can be exponentially large in both $|H|$ and $|A|$, the complexity of generating $H^d$ is customarily measured in the input and output sizes. 

According to a result in \cite{EG95}, the problem $\operatorname{DUAL}(H,G)$ can be solved in incremental quasi-polynomial time, i.e., in $\mathcal{O}(n)+m^{(o(\log⁡ (m) )}$ time, where $n=|A|$ and $m=|H|+|G|$.
Moreover, $H^d$ can be generated in incremental polynomial time (i.e. $\operatorname{DUAL}( H,G)$ can be solved in time polynomial in
$|A|, |H|$, and $|G|$) for many classes of hypergraphs, see \cite{BEGK02,BGH98,EG95}.

 The hypergraph dualization algorithm is one of various optimizations in database theory, artificial intelligence, game theory, and learning theory, to name a few. See the survey articles \cite{EG95} and \cite{BEGK02}, also the most complete recent account in \cite{H08}.  Its efficient implementation is described and tested in L. Khachiyan et al.~\cite{KBEG06}. While it was shown that the implementation achieves the same theoretical worst bound, practical experience with this implementation shows that it can be substantially faster. In particular, the code can produce, in a few hours, millions of transversals for hypergraphs with hundreds of vertices and thousands of hyper-edges. Furthermore, the experiments also indicate that the delay per transversal scales almost linearly with the number of vertices and number of hyper-edges. A more recent implementation by K. Murakami and T. Uno \cite{MU13}  not only demonstrates even better time performance, but also runs fast on large-scale inputs, for which earlier algorithms do not terminate in practical time.

\section{$D$-basis}\label{Dbas}

The idea of the $D$-basis comes from concept of the $D$-relation developed in a lattice theoretic framework. The definition of this relation goes back to the work of B. J\'onsson, A. Day, R. Freese, and J.B. Nation, which showed that this relation played a critical role in the description of the structure of free lattices. See the monograph R. Freese et al.~\cite{FJN95}. The $D$-relation plays a key role in defining the OD-graph of a finite lattice in J.B. Nation \cite{N90}, which was widely used in computer science literature. This concept was translated into the $D$-basis in the recent work \cite{ANR11}.

Essentially, the $D$-basis is a subset of the canonical direct unit basis that consists of implications $x\rto b$ (binary part), as well as  $X\rightarrow b$  with $|X|>1$,  such that whenever 
any $x\in X$  is replaced by any set $Y$ for which $x\rightarrow y$ holds for all $y \in Y$, and $Y\rightarrow x$ does not hold, the implication no longer holds. In particular, when $Y=\emptyset$, this also means that the implication $X'\rto b$ fails for every proper subset $X'\subset X$.

For each implication $X \rto b$ in the $D$-basis, the subset $X$ is called a \emph{minimal cover} for $b$, in lattice theory framework, and the $D$-relation is the binary relation on the base set of a closure system defined as follows: $bDx$ if $x \in X$, for some minimal cover $X$ for $b$. 

As we will see in section \ref{Galois}, the closure system of interest for us will be defined on the base set $A'$ of attributes of a (reduced, per discussion in section \ref{reduced}) binary table, which will also become the set of join irreducible elements of the  Galois lattice $L$. We will revisit the connection between the $D$-relation and the $D$-basis in section \ref{connect}.  

 The OD-graph of a finite closure system or a finite lattice is defined as a collection of all minimal covers, if any, for join irreducible elements of the lattice, together with the partially ordered set of all join irreducibles in the lattice. The latter is reflected in the \emph{binary} part of the $D$-basis, or any other basis for the lattice.

\begin{exm}
\end{exm}
Consider the lattice $L$ on Fig.1 in section \ref{Exm}. The set of join irreducible elements is $\Ji L = \{a_1,a_2,c_1,c_2,b\}$. The OD-graph of the lattice contains $\langle \Ji L, \leq\rangle$, where $\leq$ is inherited from the lattice, i.e., the non-trivial relations are $c_1 \leq a_1,b$ and $c_2 \leq a_2,b$. Thus, implications $b\rto c_1$, $a_1 \rto c_1$ and $a_2 \rto c_2$,$b\rto c_2$ should be either included into or follow from any set of implications defining closure system represented by $L$. 

OD-graph will also contain the minimal covers, as pairs $(X,y)$, where $X\subseteq \Ji L$, $y \in \Ji L$ and $X$ is a minimal cover for $y$. Not that $\{a_1,a_2\}$ is not a minimal cover for $b$, since $a_1\rto c_1$ (while $c_1\rto a_1$ does not hold) and $a_1$ can be replaced by $c_1$ so that $c_1a_2 \rto b$ holds . Thus, the OD-graph will have only two minimal covers for $b$: $\{a_1,c_2\}, \{a_2,c_1\}$, and one minimal cover for each of $a_1,a_2$: $\{a_2,c_1\}$ and $\{a_1,c_2\}$, respectively.

Finally, the $D$-basis is the implicational form of the OD-graph and consists of implications: $a_1 \rto c_1$, $b\rto c_1$, $a_2 \rto c_2$, $b\rto c_2$, $a_1c_2 \rto b$, $a_1c_2 \rto a_2$, $a_2c_1 \rto b$ and $a_2c_1 \rto a_1$. 

The canonical direct unit basis will have three extra implications: $a_1a_2 \rto b$, $ba_1\rto a_2$ and $ba_2\rto a_1$, which are not included in the $D$-basis.\\

An important feature of the $D$-basis is that it is \emph{ordered direct}, which means that it is iteration-free when a special ordering is imposed on the implications.  For the $D$-basis, the ordering only requires that all binary implications $x \rto b$ precede all non-binary implications $X \rto c$.
This property of the basis has an advantage of easy parsing for the processing of its logical consequences. Such processing shows faster times than the forward chaining algorithm, which is widely used in industrial logic programming, or LINCLOSURE algorithm in data bases, see \cite{ANR11}, section 7. Thus it retains the important property of \emph{directness} of the canonical direct unit basis, while having, on average, only a portion of implications from the latter. According to the test results in \cite{ANR11}, the average number of implications in the $D$-basis, for closure systems on a domain of size 7, are between 0.55 and 0.70 of the size of the canonical unit basis.

\section{Galois lattice}\label{Galois}
\subsection{Support function and concepts}

Many data sets in computer science are presented in the form of binary tables. By a finite binary table we understand a relation $R\subseteq U\times A$, for the set of objects $U$ (rows of the table) and set of attributes $A$ (columns of the table). If $r=(u,a)\in R$, then the position in row $u$  and column $a$  is marked by $1$. This can be interpreted as meaning that object $u$ possesses attribute $a$. Otherwise, the position is marked  with a $0$.

This table represents a Galois connection, and allows us to form the corresponding Galois lattice. In order to define the Galois lattice, one needs to define the support function on each of the sets $2^U$ and $2^A$ with respect to $R$.
 
$S_A: 2^A\rightarrow  2^U$ is called a support function on $2^A$ if, for every  $X\subseteq A$, $S_A (X)=\{y\in U:(x,y)\in R, \text{ for all } x\in X\}$. Similarly, the support function $S_U: 2^U\rightarrow  2^A$ is defined for all  $Y\subseteq U$ as $S_U (Y)=\{x\in A:(x,y)\in R,\text{ for all } y\in Y\}$. One may use the symbol $S$ for notation of both $S_A$ and $S_U$, since it is usually clear from the context which one should be applied.

The pair $(X,Y) \in 2^A\times 2^U$ is called a \emph{concept} of the relation $R$, if $Y=S(X)$ and $X=S(Y)$. It is easy to show that if $(X_1,Y_1)$  and $(X_2 ,Y_2)$ are two concepts with $X_1 \subseteq X_2$, then $Y_2\subseteq Y_1$. 	Thus, the set of all concepts can be ordered with respect to the set containment order on their first components, or with respect to containment order of their second component, producing two partially ordered sets that are dual to each other. 

Moreover, each of these partially ordered sets actually forms a lattice, $L_R$, or respectively, $L_R^*$, and one may refer to one or the other (depending on the preferences between sets $U$ or $A$) as the Galois lattice (\emph{concept lattice} in FCA) of the relation $R$.

Since the structure of the Galois lattice $L_R$ is fully determined by its first component (or its dual lattice $L_R^*$ is fully determined by its second component), one can establish an isomorphism between the Galois lattice $L_R$ and the lattice of closed sets of a closure operator defined on $2^A$ by means of the support function. Indeed, it is straightforward to show that the operator $\phi_A: 2^A\rightarrow 2^A$ defined as $\phi_A (X)=S_U (S_A (X))$ for $X\in 2^A$ is, in fact, a closure operator on $A$, and the closed sets with respect to $\phi_A$ are exactly the first 
components of the concepts of $R$. Thus, $\Cl(A,\phi_A)=L_R$, where $\Cl(A,\phi_A)$ denotes the lattice of closed sets of a closure system $(A,\phi_A)$.

\subsection{Reductions of the table}\label{reduced}

There is a well-developed procedure for reducing the given relation R to a relation $R'\subseteq U'\times A'$, where $U'\subseteq U$, $A'\subseteq A$ and $R'=R|_{U'\times A'}$, so that $|U'|$ and $|A'|$ are minimal with respect to property $L_R'\simeq L_R$. See, for example, section 2 of \cite{ANR11}, for the general theoretical outline, and section XI.3 in \cite{FJN95} for algorithmic details. In other words, one may leave only essential objects and attributes in the relation and remove the others. Any analysis can be done on the smaller table, while easily extending the outcome to the original sets of objects and attributes. In fact, every removed element $a\in A\setminus A'$ will be associated with some $X_a\subseteq A'$ such that $\phi_{A' } (a)=\phi_{A' } (X_a)$.
	
When the procedure of reduction is completed and one obtains the reduced table $R'$, one can establish an important relationship between the elements of the set of objects $U'$, the set of attributes $A'$, and the special subset of elements in the Galois lattice $L_R'$. Namely, the set $A'$ can be interpreted as the set of join irreducible elements  $\Ji L_R'$ of the Galois lattice $L_R'$, and $U'$ as the set of meet irreducible elements $\Mi L_R'$.\footnote{It is a matter of taste which of the two Galois lattices, dual to each other, one chooses. For that matter, the set $A'$ may be associated with the set of meet irreducibles rather than join irreducibles, which often happens in publications of FCA.} Moreover, an element $1$ at the intersection of row $i$ and column $j$ is equivalent to $j\leq i$ in the lattice $L_R'$. 

\subsection{Arrow relations in the table}\label{arr}

We may assume that after the reduction the table has rows of $n$ objects and columns of $m$ attributes. 
The binary table allows one to quickly recover additional information on $L_R'$:

\begin{itemize}
\item [(1)] Establishing a partial order $(U',\leq )=(\Mi L_R' ,\leq )$;
\item [(2)] Establishing a partial order $(A',\leq )  =(\Ji L_R' ,\leq )$;
\item [(3)] Establishing arrow relations $\uparrow$, $\downarrow$, and $\updownarrow$.
\end{itemize}
	Recall that for $i\in \Ji L_R'$ and  $j\in \Mi L_R'$, $i\uparrow j$ is defined to hold iff  $j$  is a maximal element among elements of $L_R'$ that are not greater than $i$. (In lattice terms, it is equivalent to: $i\vee j=j^*$, where $j^*$ is the unique upper cover of $j$).

Dually, $i\downarrow j$  iff  $i$  is a minimal element among elements of $L_R'$ that are not less than $j$. (This is equivalent to: $i\wedge j=i_*$, where $i_*$ is the unique lower cover of $i$). Finally, $i\updownarrow j$, if both $i\uparrow j$ and $i\downarrow j$ hold.

It is clear that algorithmically, the reconstruction of arrow relations is equivalent to finding maximal or minimal elements in a particular partially ordered set, which are sub-posets of either $(U',\leq )$ or $(A',\leq)$. So there is a straightforward process to recover these relations. Thus, we may assume that the reduced table of lattice $L_R'$ is given equipped with the arrow relations. 

\subsection{Implicational basis}

The Galois lattice can be fully determined by a set of implications defined on the set $A'$, or dually, on the set $U'$. 
By definition, an \emph{implication} on the set $A'$ is an ordered pair $(X,Y)\in 2^{A'}\times 2^{A'}$ with $X,Y \neq \emptyset$. Very often the implication $(X,Y)$ is written in the form $X\rightarrow Y$.  A subset $Z\subseteq A'$ respects  an implication $X\rightarrow Y$, if whenever $X\subseteq Z$ one also has $Y\subseteq Z$. If $\Sigma$ is a set of implications, then  Z  respects $\Sigma$ whenever Z respects every implication $\sigma$ from $\Sigma$.

There exists a classical connection between closure operators defined on $A'$ and sets of implications on $A'$:
\begin{itemize}
\item[(1)] every set of implications $\Sigma$ defines a closure operator $\phi_\Sigma$ by setting $\phi_\Sigma (Y)$ as the smallest overset of $Y$ that respects $\Sigma$;

\item[(2)] every closure operator $\phi$ can be defined by some set of implications $\Sigma$ such that the $\phi$-closed sets are exactly sets that respect $\Sigma$.
\end{itemize}

While every set of implications defines the closure operator uniquely, there are multiple possibilities to define a set of implications for any given operator.  A set of implications for a given closure operator, satisfying some conditions of minimality, is usually called an implicational basis of this closure system. One can refer to the survey article \cite{BGKM02} for further details on connections between closure operators, their closure lattices, and sets of implications.
	
For application purposes, most often the implications 
on $A'$ provide an essential hidden connection between different attributes in the given data set. For example, an implication $a_1 a_2\rightarrow b$, where $a_1 ,a_2,b\in A'$, means that every object in the data set that possesses both attributes $a_1$ and $a_2$ also possesses attribute $b$.

\subsection{Connection between $D$-basis and $D$-relation}\label{connect}

As we mentioned earlier, the name of the $D$-basis is directly connected to the $D$-relation defined in a lattice theory framework. Lemma 2.31 in \cite{FJN95} can be formulated as follows.
\begin{lem} \label{lemma2}
Given two elements $b,c\in \Ji L$, the relation $bDc$ holds iff there exists an implication  $X\rightarrow b$ in the $D$-basis of $L$ such that $c\in X$.
\end{lem}
 In particular, for every implication $X\rightarrow b$ in the D-basis, we have $X\subseteq bD=\{x\in \Ji L:bDx\}$.

\section{Recovery of the $D$-basis from the table}\label{D-bas-table}

The goal of this section is to prove the main result of the paper.
\begin{thm} Given a reduced table $(U',A',R')$, $R'\subseteq U'\times A'$, one can polynomially (in the size of the table) reduce the problem of recovery of the $D$-basis to (parallel) solution of the hypergraph dualization problem  formed for each $b \in A'$.
\end{thm}
\begin{proof}
Recall from section \ref{arr} that set $A'$ can be interpreted as the set of join irreducible elements $\Ji L$ of the Galois lattice $L$.
The key observation for the recovery of the $D$-basis is contained in \cite[Lemma 11.10]{FJN95} that connects the arrow relations of the reduced table with the $D$-relation on the set of join irreducible elements of its Galois lattice:
\[
 bDc     \text{      iff      }      b\uparrow p \text{  and  } c\downarrow p, \text{ for some } p\in \Mi L_R'.
\]
This allows us to recover effectively the sets $bD=\{c\in A':bDc\}$, for every $b\in \Ji L$.

According to Lemma \ref{lemma2} in section \ref{connect}, for every fixed $b$, every implication of the $D$-basis of the form  $X\rightarrow b$ will satisfy 
$X\subseteq bD=\{x\in \Ji L:bDx\}$. 
	
Another important subset associated with each $b$ is $M(b)=\{m\in \Mi L:b\uparrow m\}$. Recall that these are the maximal elements in $L$  which are not greater than $b$. The following statement is a simple lattice theoretical observation. 

\begin{claim}
For every $b\in \Ji L$ and $Y\subseteq \Ji L$,  $Y\rightarrow b$ holds in $L$ iff  for every $m\in M(b)$ there exists $y\in Y$ such that $y\not \leq m$. 
\end{claim}
	
This observation reduces the problem of finding the $D$-basis to the problem of finding the minimal transversal sets of a particular Sperner hypergraph.

Let first consider the general setting of an optimization problem which has a standard reduction to the hypergraph dualization problem. Given any set $X$ and family $\mathcal{M} = \{M_1, \dots M_k\} \subseteq 2^X$ of its subsets, which, we may assume, are pairwise incomparable, consider an order ideal $\mathcal{J}$ generated by this family. In other words, $\mathcal{J} = \{Y\subseteq X: Y\subseteq M_i, \text{ for some } i\leq k\}$. Apparently, the family of subsets $2^X\setminus \mathcal J = \{Z \subseteq X: Z \notin \mathcal{J}\}$ forms an order filter $\mathcal{F}$ in the poset $2^X$. A well-known optimization problem asks to find all the minimal
elements of $\mathcal{F}$, i.e., sets $Y_1,\dots, Y_s$ such that $\mathcal{F}=\{Z\subseteq X: Y_i\subseteq Z \text{ for some } i\leq s\}$.

The standard reduction to the hypergraph dualization is done by defining the hypergraph $(X, H)$, where $H=\{M_1^c,\dots , M_k^c\}$, $M_i^c = X\setminus M_i$. Apparently, any transversal $T$ of this hypergraph has at least one element from $M_i^c$ for each $i\leq k$, so that $T$ does not belong to the ideal $\mathcal{J}$ generated by $M_1,\dots, M_k$. Therefore, it belongs to $2^X \setminus \mathcal{J} = \mathcal{F}$. Vice versa, any element $Y\in \mathcal{F}$ cannot be a subset of any $M_i$, $i\leq k$, hence, $Y\cap M_i^c\not = \emptyset$. Therefore, $Y$ must be a transversal of the hypergraph $(X, H)$. In the conclusion, finding the minimal elements of $\mathcal{F}$ is equivalent of finding the minimal transversals of $(X, H)$.

We will proceed by setting an instance of optimization problem above, for each $b \in A'$.
 Let $X=bD=\{c\in A': bDc\}$. Consider family of subsets $\mathcal{M} = \{M_m=bD\cap [0,m]: m\in M(b)\}$.
Here $0$ stands for the smallest element of the lattice $L$ and $[0,m]=\{ x\in  L:x\leq m\}$. For each $m$, this is equivalent to taking only those elements from $bD$ that are in the relation $R'$ with $m$. 

Let $\mathcal{J}$ be an order ideal in $2^X=2^{bD}$ generated by a family $\mathcal{M}$. 
According to the Claim, we need to find subsets $Y\subseteq bD$ that do not belong to $\mathcal{J}$. Thus, finding the minimal such sets $Y$ is equivalent to an instance of the optimization problem we discussed above.

Note that, in comparison with the method of \cite{RDB11}, we build a hypergraph on $bD$ rather than whole $A'$, which gives the reduction of the hypergraph size.

 The hypergraph dualization problem would search for the minimal elements of $\mathcal{F}=2^{bD}\setminus \mathcal{J}$, where $\mathcal{J}$ is the order ideal in  $2^{bD}$ generated by the family $\mathcal{M}$.

If $Y_1,Y_2\dots ,Y_s$ are such minimal elements, then the set of implications  $Y_i\rightarrow b$, $i=1,\dots k$, gives a set of implications satisfying two properties:
\begin{itemize}
\item
$Y_i$  is a minimal subset X such that $X\rightarrow b$ holds;
\item
$Y_i  \subseteq bD$.
\end{itemize}

In particular, all these implications belong to the canonical unit basis, and every implication from the $D$-basis of the form $X\rightarrow b$ is included into this list. The full $D$-basis will be recovered, when this process is applied to all $bD\not = \emptyset$, $b\in A'$. 

As a result, the problem to recover the $D$-basis is polynomially reduced to at most $t$ runs of the dualization algorithm on $2^{X_i}$, $X_i\subseteq A'$, where $t=|A'|$ is the number of the attributes.
It is possible, however, that some of the recovered implications are not in the $D$-basis, so the recovered set of implications may contain the $D$-basis properly. 
\end{proof}

\section{Example}\label{Exm}

Let us consider the procedure on a small table with 6 objects and 7 attributes. All symbols in the table other than 1 should be first interpreted as 0s.

\begin{table} [htb]
\centering
\begin{tabular}{|c|c|c|c|c|c|c|c|}
\hline
   & $b$ & $a_1$ & $a_2$ & $c_1$ & $c_2$ & $u$ & $v$  
\rule[-1ex]{0ex}{3.5ex} \\
\hline
1 & $\ua$ & 1 & $\ua$ & 1 & $\uda$ & 1 & 0 
\rule[-1ex]{0ex}{3.5ex} \\
\hline
2 & 1 & $\uda$ & $\uda$ & 1 & 1 & 1 & 0 
\rule[-1ex]{0ex}{3.5ex} \\
\hline
3 & $\ua$ & $\ua$ & 1 & $\uda$ & 1 & 0 & 0
\rule[-1ex]{0ex}{3.5ex} \\
\hline
4 & $\uda$ & $\da$ & $\da$ & 1 & 1 & 1 & 0
\rule[-1ex]{0ex}{3.5ex} \\
\hline
5 & 0 & 0 & 0 & 1& 1 & 1 & 0
\rule[-1ex]{0ex}{3.5ex} \\
\hline
6 & 1 & 1 & 1 & 1 & 1 & 1 & 1
\rule[-1ex]{0ex}{3.5ex} \\
\hline
\end{tabular}
\end{table}

Here $U=\{1,2,3,4,5,6\}$, $A=\{b,a_1,a_2,c_1,c_2,u,v\}$. Since $S(c_1 )=S(u)$, attribute $u$ can be reduced. Attribute $v$ also can be reduced, because $S(S(v))=A$, and $A\setminus v$ is not a first component of any concept. One can also reduce object 5 due to $S(5)=S(4)$, and reduce object 6 because $6\in S(S(i))$, for every $i\in \{1,2,3,4,5\}$. Thus, one can consider the reduced table with $U'=\{1,2,3,4\}$, $A'=\{b,a_1,a_2,c_1,c_2 \}$, and complement the basis (on $A$) of a reduced table by implications $c_1\rightarrow u$, $u\rightarrow c_1$, $v\rightarrow A\setminus v$.

	The order relation on $A'$  consists of $c_1\leq a_1, b$ and $c_2\leq a_2, b$. The order relation on $U'$ consists of $2\geq 4$. This allows us to determine the arrow relations between elements of $A'$ and $U'$, which are shown in the table. The Galois lattice $L_R'$ of the reduced table in shown on Fig.~1. Note that objects of the reduced table correspond to the following (meet irreducible) elements of the lattice: 
$1=a_1$, $2=b$, $3=a_2$ and $4=c_1\vee c_2$.

\begin{figure}[ht]\label{L4}
\includegraphics[height=2.0in,width=6.0in]{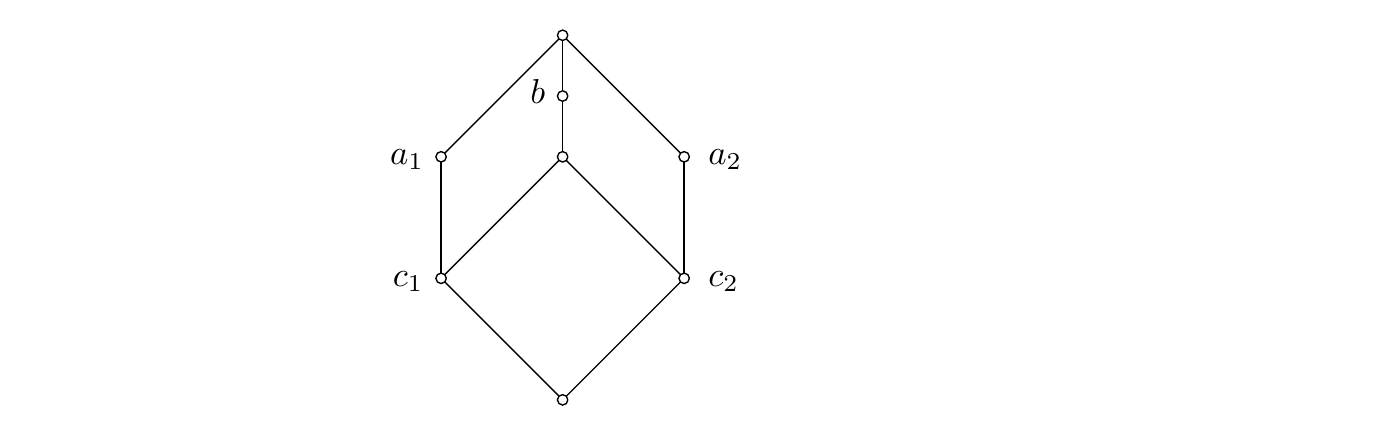}
\caption{Galois lattice from the table}
\end{figure}

In order to recover all implications $X\rightarrow b$, for some $X\subseteq A'$, first identify elements of $bD=\{a_1,a_2,c_1,c_2\}$. Also, $M(b)=\{1,3,4\}$. Taking any $m\in M(b)$, we find corresponding subsets of family $\mathcal{M}$: $M_1=\{a_1,c_1 \}$, $M_3=\{a_2,c_2 \}$, $M_4=\{c_1,c_2 \}$. This is done by picking the support of each element $m\in M(b)$ within $bD$. The hypergraph dualization problem finds the minimal elements in $2^{bD}\setminus J$, where $J$ is the order ideal in  $2^{bD}$ generated by the family $\mathcal{M}=\{M_1,M_3,M_4\}$. Evidently, these will be $Y_1=\{a_1,c_2\}, Y_2=\{a_2,c_1\}$, and $Y_3=\{a_1,a_2\}$. This gives all implications from the basis with the conclusion $b$: $a_1 c_2\rightarrow b$, $a_2 c_1\rightarrow b$ and $a_1 a_2\rightarrow b$.
	
We note that the last implication is not in the $D$-basis and one can use the algorithm from \cite{ANR11}  to remove it.

We can mention that the reduction of the retrieved basis can be considerable, and the size of the reduced part may depend on the number of \emph{binary} implications in the basis, i.e., implications of the form $x \rto y$.
For example, one of the tests on a small matrix of size $10\times 22$ found 9 binary implications in the basis, and the $D$-basis had a total of 635 implications, after 212 implications were reduced in the last phase of the algorithm. No reduction will occur if the binary part of the basis is empty: in this rare case the canonical direct basis and $D$-basis coincide. 

\section{Results of initial testing}\label{test}

The algorithm of $D$-basis recovery described in section \ref{D-bas-table} was implemented in C++ code by the team of undergraduate students of Yeshiva University. They were able to implement the call to an existing subroutine that performs the hypergraph dualization, and which is publicly available via repository maintained by T. Uno: http://research.nii.ac.jp/$\sim$uno/dualization.html 

Further optimizations of the code were done in collaboration with T. Uno and U. Norbisrath. The retrieval of the full implicational basis for a random matrix $20\times 40$ and density 0.2 (20\% of entries of the matrix are ones) took 0.27 sec, for 1616 implications. A similar test in U. Russel et al.~\cite{RDB11}  mentioned in section \ref{FCA}, showed 0.9 sec for 1476 implications.

Further tests with the new code were done on larger random matrices. On a table of size 50-by-100, the $D$-basis with 49,000 implications was obtained in 3 min 30 sec. All the implications $X\rightarrow b$, for a requested column $b$ of a randomly generated matrix of size 50-by-200, were obtained in 25 min.

For the initial comparisons with \emph{Apriori} algorithm, we ran two types of tests. 

The first data set was taken from the Frequent Itemset Mining Dataset Repository, publicly available at http://fimi.ua.ac.be/data/retail.dat. It was retail market basket data from an anonymous Belgian retail store. We took first 90 rows converting them to a binary matrix format with size 90-by-502 and  (low) density 0.0162. Running time was about 42 sec resulting in 104 mostly binary implications of maximal support 5 and confidence = 1. For the \emph{Apriori}, we used Microsoft SQL Server Business Intelligence Development Studio, 2008 (Data Mining Technique - Microsoft Association Rules) \cite{JK06, ZM05}. The result of this run, together with setup of the input, took about 4 min 30 sec, with a considerably larger set of association rules, most of which have confidence $< 1$. 
We point that Apriori was designed specifically for mining association rules in retail data, and the specifically for data that is presented by large sets of item-sets. When converted to matrix form, it usually contains large number of rows (transactions) with comparably few columns (number of items at sale), and normally has low density.


Our second data set was of a different nature, where we believe our novel approach may have an edge over \emph{Apriori}.

 The data-set was kindly provided by the research group of Dr.~G. Okimoto from the University of Hawaii Cancer Center. It contained the gene expression levels for 550 pre-selected genes, in 22 patients, some healthy and others with liver cancer. There are multiple approaches how to convert this matrix into binary format. For comparison of $D$-basis algorithm with Apriori, one of them was taken as test data of size 22-by-1112 and density 0.4684. Columns 1111 and 1112 represent the attributes of being healthy and having cancer, correspondingly. 

Microsoft Association Rules output was restricted to only 2000 frequent sets with minimum support 15 and confidence = 1, and there were no association rules with attributes 1111 and 1112.
On the other hand, with the $D$-basis code we were able to reveal the equivalence of attribute 1111 to the set of 9 other attributes, and to find 14819 implications with the target attribute 1112, whose support = 8 and the confidence at least 21/22.

\section{Further theoretical work and implementation}\label{future}

Recovery of the implicational basis directly relates to knowledge discovery, especially if one is concerned with the targeted set of implications in the basis. Additional testing of biological data provided by the collaborators at the University of Hawaii Cancer Center and medical group in Astana, Kazakhstan, were aiming at data of a larger size, and at the retrieval of specific targeted association rules. The results of this testing were recently presented at the FCA conference \cite{ANO15}.

From a theoretical point of view, the future plans include the generalization of the algorithm for the retrieval of association rules with the threshold of confidence $\mu <1$ (see section \ref{Asso-rul}), and here some initial results are achieved, which may be implemented in the code. This will allow us to take into account the possibility of incomplete and erroneous data. On the other hand, a different approach exists which allows us to utilize the existing implementation. Namely, the algorithm can be run multiple times on the proper subsets of existing sets of rows. For example, if the original data contains 100 rows, then the algorithm can be run on all subsets of 95 rows, which will result in obtaining all association rules of confidence at least $95\%$. This can be achieved easily with the assistance of a specialist in distributed computing.

As far as the goal of the current paper, theoretical results of the $D$-basis recovery are supported by the practical evidence tested in \cite{RDB11} and initial testing presented in section \ref{test}, that the recovery of the implicational basis based on the hypergraph dualization algorithm will bring a considerable cut in run-times when dealing with tabled data. This will make the new approach critical for analysis of large data sets.

\emph{Acknowledgments.} The results of this paper were prompted by the discussion of the hypergraph dualization algorithm with E. Boros and V. Gurvich, at the RUTCOR seminar during the first author's visit in 2011. The test results of section \ref{test} were possible due to code implementation done by J. Blumenkopf and T. Moldwin (Yeshiva College, New York), assistance of T. Uno (National Institute of Informatics, Tokyo), U. Norbisrath and A. Amanbekkyzy (Nazarbayev University, Astana), and data provided by G. Okimoto (University of Hawaii, Honolulu).

\end{document}